\pdfoutput=1
\documentclass[format=acmsmall,natbib,nonacm]{acmart}

\copyrightyear{2020}
\copyrightyear{2020}
\acmYear{2020}
\setcopyright{acmlicensed}
\PassOptionsToPackage{dvipsnames}{xcolor}

\usepackage{algorithm}
\usepackage{algorithmicx}
\usepackage{algpseudocode}
\usepackage{graphicx}
\usepackage{amsmath}
\usepackage{multirow}
\usepackage{hyperref}
\usepackage{cleveref}
\usepackage{natbib}
\usepackage{paralist}
\usepackage[T1]{fontenc}
\usepackage{beramono}
\usepackage{listings}
\usepackage[dvipsnames]{xcolor}
\usepackage[inline]{enumitem}
\usepackage{acronym}
\usepackage{bm}
\usepackage[skip=0pt]{caption}
\usepackage[skip=0pt]{subcaption}
\usepackage{bbm}
\usepackage{xifthen}
\usepackage{amsthm}
\usepackage{siunitx}
\usepackage{booktabs}
\usepackage[np, autolanguage]{numprint}
\usepackage{mathtools}

\acrodef{IR}{Information Retrieval}
\acrodef{LTR}{Learning to Rank}
\acrodef{PL}{Plackett-Luce}
\acrodef{SGD}{Stochastic Gradient Descent}
\acrodef{CRM}{Counterfactual Risk Minimization}
\acrodef{IPS}{Inverse Propensity Scoring}
\acrodef{CH}{Chernoff-Hoeffding}
\acrodef{MPeB}{Maurer \& Pontil empirical Bernstein}
\acrodef{SEA}{Safe Exploration Algorithm}
\acrodef{BSEA}{Boundless Safe Exploration Algorithm}
\acrodef{DBGD}{Dueling Bandit Gradient Descent}

\theoremstyle{definition}

\newcommand{\iteration}{t}
\newcommand{\nriterations}{T}

\newcommand{\contextvector}{\bm{x}}

\newcommand{\contextdimensions}{m}

\newcommand{\action}{a}

\newcommand{\nractions}{n}

\newcommand{\reward}{r}

\newcommand{\propensity}{p}

\newcommand{\lcb}{LCB}
\newcommand{\ucb}{UCB}

\newcommand{\dataset}{\mathcal{D}}

\newcommand{\policy}{\pi}

\newcommand{\labels}{y}

\newcommand{\real}{\mathbb{R}}

\newcommand{\datasetrcv}{RCV1}
\newcommand{\datasetnews}{20 Newsgroups}
\newcommand{\datasetusps}{USPS}

\newcommand{\rqone}{Is \ac{SEA} safe? I.e., does it always perform at least as good as a baseline policy?}
\newcommand{\rqtwo}{Does \ac{SEA} provide a better user experience during training than offline or online learning methods? I.e., does it accumulate a higher reward during training?}
\newcommand{\rqthree}{Does \ac{SEA}, which explores the action space, learn a more effective policy compared to offline learning methods, which do not perform exploration?}

\newcommand\doSingleLine[1]
{\let\normalnewline=\\
	\let\\=\relax
	#1%
	\let\\=\normalnewline}

\begin{document}

\title{Safe Exploration for Optimizing Contextual Bandits}

\author{Rolf Jagerman}
\orcid{0000-0002-5169-495X}
\affiliation{%
	\institution{University of Amsterdam}
	\streetaddress{Science Park 904}
	\postcode{1098 XH} 
	 \city{Amsterdam}
	 \country{The Netherlands}
}
\email{rolf.jagerman@uva.nl}

\author{Ilya Markov}
\affiliation{%
	\institution{University of Amsterdam}
	\streetaddress{Science Park 904}
	\postcode{1098 XH}
	\city{Amsterdam}
	\country{The Netherlands}
}
\email{i.markov@uva.nl}
 
\author{Maarten de Rijke}
\orcid{0000-0002-1086-0202}
\affiliation{
	\institution{University of Amsterdam}
	\streetaddress{Science Park 904}
	\postcode{1098 XH}
	\city{Amsterdam}
	\country{The Netherlands}
}
\email{derijke@uva.nl}


\begin{abstract}
Contextual bandit problems are a natural fit for many information retrieval tasks, such as learning to rank, text classification, recommendation, etc.
However, existing learning methods for contextual bandit problems have one of two drawbacks: they either do not explore the space of all possible document rankings (i.e., actions) and, thus, may miss the optimal ranking, or they present suboptimal rankings to a user and, thus, may harm the user experience.
We introduce a new learning method for contextual bandit problems, \ac{SEA}, which overcomes the above drawbacks.
\ac{SEA} starts by using a baseline (or production) ranking system (i.e., policy), which does not harm the user experience and, thus, is safe to execute, but has suboptimal performance and, thus, needs to be improved.
Then \ac{SEA} uses counterfactual learning to learn a new policy based on the behavior of the baseline policy.
\ac{SEA} also uses high-confidence off-policy evaluation to estimate the performance of the newly learned policy.
Once the performance of the newly learned policy is at least as good as the performance of the baseline policy,
\ac{SEA} starts using the new policy to execute new actions, allowing it to actively explore favorable regions of the action space.
This way, \ac{SEA} never performs worse than the baseline policy and, thus, does not harm the user experience,
while still exploring the action space and, thus, being able to find an optimal policy.
Our experiments using text classification and document retrieval confirm the above by comparing \ac{SEA} (and a boundless variant called \acs{BSEA}) to online and offline learning methods for contextual bandit problems.
\end{abstract}

        

\maketitle
	

\section{Introduction}
\label{section:introduction}

A \emph{multi-armed bandit} problem is a problem in which a limited number of resources must be allocated between alternative choices so as to maximize their expected gain.
In \emph{contextual} bandit problems, when making a choice, a representation of the context is available to inform the decision.
Contextual bandit problems are a natural framework to capture a range of \ac{IR} tasks~\citep{hofmann-contextual-2011,adomavicius2005incorporating,glowacka-2019-bandit}. 
Example tasks to which contextual bandits have been applied include news recommendation~\citep{li2010contextual}, text classification~\cite{swaminathan2015counterfactual}, ad placement~\citep{langford2008exploration}, and online learning to rank~\citep{hofmann-balancing-2011}.
For example, in online \ac{LTR} (see Figure~\ref{fig:rankingasbandits}):
\begin{enumerate}
	\item a user issues a query,
	\item a search engine presents a ranked document list, and
	\item clicks on the documents are recorded as feedback.
\end{enumerate}
The interactive nature of this problem makes it an ideal application for the contextual bandit framework. 
The \ac{LTR} model is a decision-making \emph{policy} $\policy$, the ranked document lists it produces are \emph{actions}, the user is the \emph{context} or \emph{environment}, and clicks are the \emph{reward} signal.

There are two major classes of algorithms to maximize the reward of a contextual bandit policy~$\policy$.
The first are online algorithms, which optimize a policy while it is being executed~\cite{li2010contextual,agrawal2013thompson,kaelbling1994associative,wei2017reinforcement}.
The second are offline algorithms, which optimize a policy based on data that was collected using an existing logging policy, often called $\policy_0$ \cite{strehl2010learning,swaminathan2015counterfactual}.

\begin{figure}
	\includegraphics[width=.6\columnwidth]{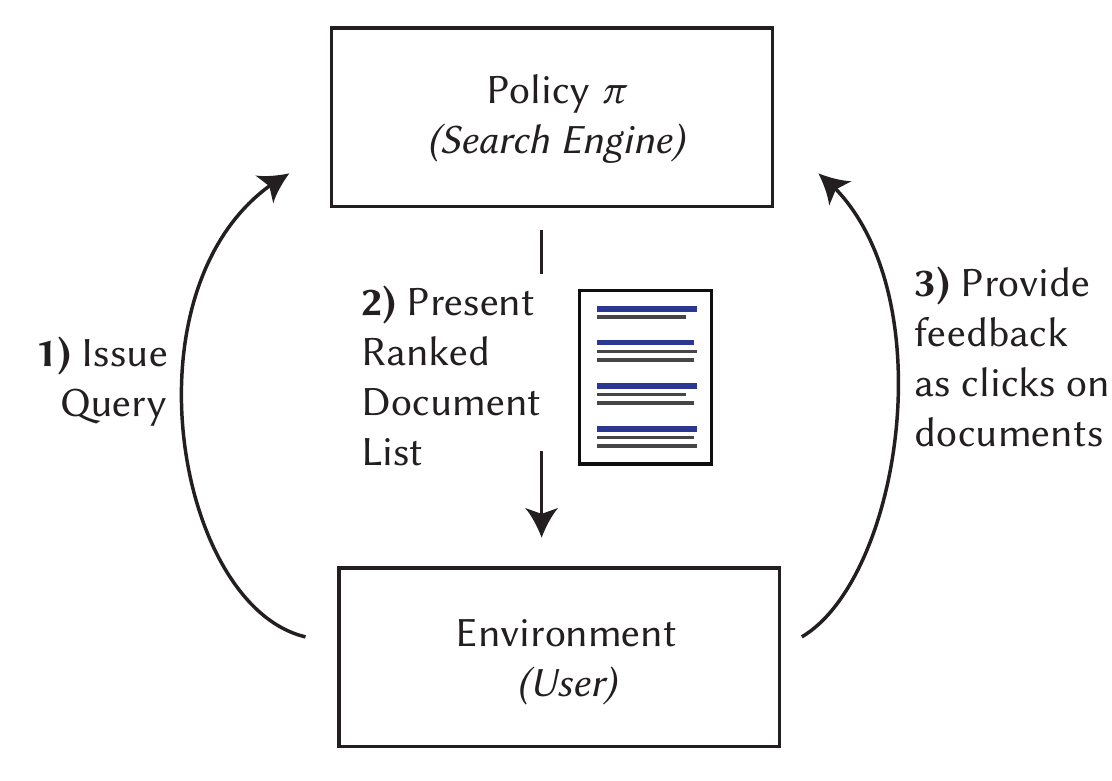}
	\caption{Online learning to rank viewed as a contextual bandit problem. A user issues a query, the search engine responds with a ranked list of documents, and the user provides implicit feedback in the form of clicks on documents.}
	\label{fig:rankingasbandits}
\end{figure}

\textit{Online} learning methods for contextual bandits are widely studied and there are many known algorithms to solve this problem.
E.g., popular algorithms include $\epsilon$-greedy, LinUCB~\cite{li2010contextual} and Thompson Sampling~\cite{agrawal2013thompson}.
Despite their attractive properties, the adoption of online learning methods for contextual bandits in production systems has been limited.
Especially in the early stages of learning, online algorithms may perform actions that are suboptimal and, thus, hurt the user experience.
E.g., in online \ac{LTR} it is risky to present suboptimal rankings of documents~\citep{oosterhuis-balancing-2017,wang-position-2018}.

\emph{Offline} learning from logged feedback~\cite{strehl2010learning,hofmann-reusing-2013,swaminathan2015counterfactual,grotov-bayesian-2015} has been proposed as a solution to this problem. 
Using an existing and already deployed logging policy, one only takes actions from that logging policy to collect bandit feedback.
Using this collected data, a new policy is then learned offline in an unbiased manner. 
This learning process is safe in the sense that the new policy is not executed and, thus, the user experience is not affected.
The drawback of offline learning, however, is that it relies on the deployed logging policy, which never changes by design.
Due to this, there may be areas in the action space left unexplored, i.e., bandit feedback may be missing for those areas.
In \ac{IR} terms, potentially high quality document rankings might never be presented to users
and the corresponding user interactions will never be observed.

In this paper, we propose a method that performs \textit{exploration} of the action space in contextual bandit problems, but does this exploration \textit{safely}, that is, without hurting the user experience.
In our proposed solution, called \acf{SEA}, we use an already deployed logging policy which is safe, static but suboptimal, as a \emph{warm-start} for learning a new policy.
As soon as \ac{SEA} is confident that the performance of the new policy does not fall below that of the logging policy, it starts executing actions from the new policy, thus, exploring the action space.
In the context of \ac{LTR}, this means that \ac{SEA} is capable of exploring and presenting new rankings to users, gathering feedback for rankings that might otherwise have never been presented.
This enables \ac{SEA} to trade off various strengths and weaknesses of both \emph{online} and \emph{offline} learning.

We note that periodic deployment of a policy, after a successfull evaluation on held-out data, is a common practice in industry.
In our experiments we include a boundless version of \ac{SEA}, called \acs{BSEA}, which represents such a deployment pipeline.
What \ac{SEA} adds over and above a standard deployment procedure is:
\begin{inparaenum}
\item it comes with a formal proof of safety, i.e., with high probability the performance of \ac{SEA} will be at least as good as that of a baseline policy (Section~\ref{sec:proof}); and
\item we introduce a computationally efficient manner for computing high-confidence off-policy bounds (Section~\ref{sec:efficient}).
\end{inparaenum}

The research questions we address in this paper are:
\begin{enumerate}[label=(RQ\arabic*),align=left]
	\item \rqone\ 
	\item \rqtwo\ 
	\item \rqthree
\end{enumerate}

Our key technical contributions in this paper are:
\begin{inparaenum}
	\item we introduce \ac{SEA} and show it to be safe: its performance never falls below that of a baseline policy;
	\item we show that SEA improves the user experience: it provides higher cumulative reward during training than both offline and online methods; and
	\item we show that a policy learned with \ac{SEA}, which is capable of exploring new actions, outperforms that of offline methods, which are incapable of exploration.
\end{inparaenum}


\section{Background}

\subsection{Contextual bandits}
\label{sec:contextualbandits}

In online \ac{LTR} we try to optimize the parameters of a ranking model such that it places relevant items at the top of the ranked list. 
We can view the ranking model as a decision-making policy. 
When a new query arrives, the policy takes an action: it displays one possible ranked list to the user. 
The user then decides to click on documents shown in the ranked list, thus providing a reward signal. 
More formally, we consider the following contextual bandit framework. 
At each round $\iteration$:
\begin{enumerate}
	\item The environment announces an $\contextdimensions$-dimensional \emph{context vector}~$\contextvector_\iteration \in \mathcal{X}$. In \ac{IR} terms, this would be a user (environment) issuing a query (context vector).
	\item A policy $\policy$ samples an action $\action_\iteration \in \mathcal{A}$, one of $\nractions$ possible actions,  conditioned on $\contextvector_\iteration$: $a_t \sim \policy(\cdot \mid \bm{x}_t)$. In \ac{IR} terms, this would be a ranking model (policy) producing a ranked list (action).
	\item The environment announces only the reward $r_{t,a_t}$ for the chosen action $a_t$, and not for other possible actions that the policy could have taken. We assume that $r \in [0, 1]$. In \ac{IR} terms, this reward signal could be clicks on documents.
\end{enumerate}
This learning setup is inherently different from supervised learning, where rewards for all possible actions are known. 
We only observe rewards for actions that the policy has taken. 
The setting is referred to as \emph{partial-label problem}~\cite{kakade2008efficient}, \emph{associative bandit problem}~\cite{strehl2006experience} or \emph{associative reinforcement learning}~\cite{kaelbling1994associative}. 
In the context of \ac{LTR}, this means that we do not know the optimal ranked list but can only observe a reward signal for rankings that our policy chooses to display.

There are two major classes of algorithms for learning a contextual bandit policy.
The first are online algorithms that optimize the policy while it is being executed~\cite{li2010contextual,agrawal2013thompson,kaelbling1994associative,wei2017reinforcement}.
Algorithms of this class explore the space of possible actions but may harm the user experience, because suboptimal actions could be executed.
The second are offline algorithms, which optimize a policy based on data that was collected using an existing logging policy~\cite{strehl2010learning,swaminathan2015counterfactual}.
Algorithms of this class are safe as the newly learned policy is not executed and, thus, there is no risk of hurting the user experience.
However, offline algorithms are incapable of exploring the action space, which may harm the performance of the learned model.
In this work we build on the second class of algorithms; see below.

\subsection{Offline~learning~from~logged~bandit feedback}
\label{sec:learning}
Offline learning algorithms collect bandit feedback using an existing logging policy, which we call a \emph{baseline policy} $\policy_b$ (also referred to as $\policy_0$ in the literature~\cite{swaminathan2015counterfactual,swaminathan2016off,joachimsdeep}).
This feedback is collected in the following form, where at time $t$ we have:
\begin{equation}
\dataset_t = \left\{\left(\bm{x}_i, a_i, r_{i,a_i}, p_i\right)\right\}_{i=1}^t,
\end{equation}
where $\bm{x}_i$ is the observed context vector, $a_i$ is the action taken by the baseline policy, $r_{i,a_i} \in [0, 1]$ is the reward given by the environment, and $p_i$ is the propensity score for action $a_i$.
The propensity score is the probability of the baseline policy taking the logged action, that is, $p_i = \policy_b(a_i \mid \bm{x}_i)$~\cite{swaminathan2015counterfactual}.

To learn a new policy $\policy_w$ from the collected bandit feedback, the following maximization problem has to be solved~\cite{joachimsdeep}:
\begin{equation}
\widehat{\policy}_w = \max_{\policy_w} \frac{1}{t} \sum_{i=1}^{t} \frac{r_{i,a_i}}{p_i} \policy_w(a_i \mid \bm{x}_i).
\end{equation}
This optimization problem is solved via \ac{SGD}:
if a new tuple $(\bm{x}_i$, $a_i$, $r_{i,a_i}$, $p_i)$ is observed, we weigh the derivative of $\policy_w(a_i \mid \bm{x}_i)$ by $\frac{r_{i,a_i}}{p_i}$ and update the weights of $\policy_w$ using \ac{SGD}.
We refer to this offline learning method as \ac{IPS}~\cite{joachimsdeep}.

\subsection{High-confidence off-policy evaluation}
\label{sec:offpolicy}
In an offline learning setting, a newly learned policy $\policy_w$ is never executed, so its performance cannot be measured directly.
Instead, we \emph{estimate} its performance using the collected bandit feedback $\dataset_t$ at time $t$. The estimated reward of a policy $\policy$ can be written as~\cite{li2011unbiased}:
\begin{equation}
\hat{R}(\pi, \dataset_t) = \frac{1}{t} \sum_{i=1}^{t} \hat{R}_i = \frac{1}{t} \sum_{i=1}^{t} \frac{r_{i,a_i}}{p_i}\policy(a_i \mid \bm{x}_i) .
\end{equation}
This estimate suffers from high variance, especially when the propensity scores are small.
To resolve this issue, \citet{thomas2015high} propose several high-confidence off-policy estimators.
These estimators first calculate a confidence interval around the policy's performance
and then use the lower bound on this performance for off-policy evaluation.
The high-confi\-dence off-policy estimator that we use is based on the Maurer \& Pontill empirical Bernstein inequality.
The confidence bound can be written as~\cite{thomas2015high}:
\begin{equation}
\mathit{CB}(\pi, \mathcal{D}_t) = \frac{7b \ln\left(\frac{2}{\delta}\right)}{3(t-1)} + \frac{1}{t}\sqrt{\frac{\ln\left(\frac{2}{\delta}\right)}{t-1} \sum_{i,j=1}^t (\hat{R}_i - \hat{R}_j)^2},
\end{equation}
where $(1 - \delta) \in [0, 1]$ is the confidence level and $b$ is an upper bound on $\hat{R}$.
When both $r$ and $p$ are bounded, $b$ can be calculated exactly: $b = \frac{\max r}{\min p}$.
The lower confidence bound on the performance of a policy at time $t$, $LCB(\policy, \dataset_t)$, can be computed as:
\begin{equation}
\mathit{LCB}(\pi, \dataset_t) = \frac{1}{t} \sum_{i=1}^{t} \hat{R}_i - CB(\pi, \mathcal{D}_t).
\label{eq:ch}
\end{equation}
Similarly, we compute the \emph{upper} confidence bound on the performance of a policy as follows:
\begin{equation}
\mathit{UCB}(\pi, \dataset_t) = \frac{1}{t} \sum_{i=1}^{t} \hat{R}_i + CB(\pi, \dataset_t).
\label{eq:chupper}
\end{equation}

\noindent
The estimator described in this section provides an effective way to compute a lower and upper bound on the performance of a policy without executing it. 

In this paper we build on the estimators designed by~\citet{thomas2015high}, but provide the following contributions:
First, we develop \ac{SEA}, an algorithm for automatic safe exploration by deploying new models (Section~\ref{sec:walkthrough});
second, we provide an efficient way to compute the high-confidence bounds (Section~\ref{sec:efficient});
and, third, we formally prove that \ac{SEA} is indeed safe (Section~\ref{sec:proof}).


\section{Safe Exploration Algorithm (SEA)}

We define the notion of \emph{safety} as part of the learning process (Section~\ref{sec:safety}),
walk through the \acf{SEA} (Section~\ref{sec:walkthrough}), 
address the problem of efficient off-policy evaluation (Section~\ref{sec:efficient}), 
formally prove the safety of \ac{SEA} (Section~\ref{sec:proof}), 
and then analyse \ac{SEA} (Section~\ref{sec:analysis}).

\subsection{Safety}
\label{sec:safety}

We use the definition of \emph{safety} from so-called conservative methods~\cite{wu2016conservative,kazerouni2016conservative}.
These methods make a key assumption: there exists a baseline policy $\policy_b$ whose actions can always be executed without risk.
This assumption is very reasonable in practice.
In most industrial settings there exists a production system that we can consider to be the baseline policy.
A learning algorithm is considered \emph{safe} if its performance at any round $\iteration$ is at least as good as the baseline policy.
Thus, safety is a concept that is always measured relative to a baseline.

More formally, let us first consider the notion of regret at round $t$ during training. 
We define the \emph{regret} at time $t$ as the cumulative difference in reward obtained by executing actions from our policy $\policy$ compared to a perfect policy $\policy^*$, one that always chooses the action with maximum reward.
For notational simplicity, we denote the action that a policy chooses in response to a context vector $\bm{x}$ as $\policy(\bm{x})$:
\begin{equation}
\textit{Regret}_t(\pi) = \sum_{i=1}^t (r_{i, \policy^*(\bm{x}_i)} - r_{i,\policy(\bm{x}_i)}).
\end{equation}
A policy $\policy$ is considered \emph{safe} if its regret is always at most as large as that of the safe baseline policy $\policy_b$, i.e., for every $t = 1, \ldots, T$:
\begin{flalign}
\textit{Regret}_t(\pi) &\leq \textit{Regret}_t(\pi_b) \\
\shortintertext{or}
\sum_{i=1}^t (r_{i, \policy^*(\bm{x}_i)} - r_{i,\policy(\bm{x}_i)}) &\leq \sum_{i=1}^t (r_{i, \policy^*(\bm{x}_i)} - r_{i,\policy_b(\bm{x}_i)}) \\
\shortintertext{so}
\frac{1}{t}\sum_{i=1}^t r_{i,\policy(\bm{x}_i)} &\geq \frac{1}{t}\sum_{i=1}^t r_{i,\policy_b(\bm{x}_i)} .
\end{flalign}
In other words, at every time $t$, we want the average reward of our policy $\policy$ to be at least as large as the average reward of the baseline policy $\policy_b$.
In practice, we cannot observe the average reward of a policy without executing it.
This is problematic because we cannot know if a policy is safe until we execute it.
Fortunately, the off-policy estimators described in Section~\ref{sec:offpolicy} provide a way to \emph{estimate} the performance of a policy without executing it.

\subsection{A walkthrough of \ac{SEA}}
\label{sec:walkthrough}

The \acf{SEA} learns a new policy $\policy_w$ \emph{offline} from the output of a baseline policy $\policy_b$ using offline learning techniques described in Section~\ref{sec:learning}.
At each iteration $\iteration$, \ac{SEA} estimates the performance of the newly learned policy $\policy_{w_t}$ and the currently deployed policy $\policy_d$ using the high-confidence off-policy evaluators described in Section~\ref{sec:offpolicy}.
The new policy is only deployed online when its estimated performance is above that of the existing deployed policy.
This allows the newly learned policy to take over and start exploring.
This only happens once \ac{SEA} is confident enough that the new policy's performance will be satisfactory and safe to execute.

\begin{algorithm}
\smallskip
	\caption{Safe Exploration Algorithm (SEA)}
	\begin{algorithmic}[1]
		\State $\policy_b$ \hfill \textit{\small // Baseline policy (current production system) } \label{alg:exploration:policyb}
		\State $\policy_{w_0} \leftarrow \policy_b$ \hfill \textit{\small // Policy to be learned (initialized with baseline weights)}  \label{alg:exploration:policyw}
		\State $\policy_d \leftarrow \policy_b$ \hfill \textit{\small // The deployed policy that is executing actions \label{alg:exploration:policy}}
		\State $\mathcal{D}_0 \leftarrow \{\}$
		\For{$\iteration = 1 \ldots T$}
		\State $\contextvector_\iteration \leftarrow $ contextual feature vector at time $\iteration$ \label{alg:exploration:context}
		
		\State $\action_\iteration \sim \policy_d(\cdot \mid \contextvector_\iteration)$ \label{alg:exploration:draw}
		\State $\propensity_\iteration \leftarrow \policy_d(\action_\iteration \mid \contextvector_\iteration)$ \label{alg:exploration:propensity}
		
		\State Play $\action_\iteration$ and observe reward $r_{t,a_t}$ \label{alg:exploration:play}
		\State $\mathcal{D}_t \leftarrow \mathcal{D}_{t-1} \cup \{(\bm{x}_t, a_t, r_{t,a_t}, p_t)\}$
		\State $w_{t} \leftarrow w_{t - 1} + \eta \frac{r_{t,a_t}}{p_t} \nabla_w \pi_{w_{t-1}}(\action_\iteration \mid \contextvector_\iteration) $ \hfill \textit{\small // Update weights via gradient ascent}  \label{alg:exploration:sgd}
		\State Compute $\lcb(\policy_{w_t}, \mathcal{D}_t)$ using Equation~\ref{eq:ch} \label{alg:exploration:newlcb}
		\State Compute $\ucb(\policy_d, \mathcal{D}_t)$ using Equation~\ref{eq:chupper} \label{alg:exploration:baselinelcb}
		
		\If{$\lcb(\policy_{w_t}, \mathcal{D}_t) \geq \ucb(\policy_d, \mathcal{D}_t)$} \label{alg:exploration:comparison} 
		\State $\policy_d \leftarrow \policy_{w_t}$ \hfill \textit{\small // Deploy new policy only when it is safe to do so}
		\EndIf
		\EndFor
	\end{algorithmic}
	\label{alg:exploration}
\smallskip	
\end{algorithm}

The pseudocode for the \acf{SEA} is provided in Algorithm~\ref{alg:exploration}.
\ac{SEA} starts from a baseline policy $\policy_b$ (Line~\ref{alg:exploration:policyb}) and a new policy that we wish to optimize, $\policy_{w_0}$ (Line~\ref{alg:exploration:policyw}).
At the start, \ac{SEA} deploys the policy $\policy_d$, which is initialized to the baseline policy (Line~\ref{alg:exploration:policy}).
At every iteration $\iteration$, a context vector $\contextvector_\iteration$ is observed (Line~\ref{alg:exploration:context}).
\ac{SEA} draws an action from the deployed policy and computes its propensity score (Lines~\ref{alg:exploration:draw}--\ref{alg:exploration:propensity}).
\ac{SEA} executes the chosen action $\action_\iteration$ and observes the reward $\reward_\iteration$ (Line~\ref{alg:exploration:play}).
The policy $\policy_w$ is then updated via \ac{SGD} (Line~\ref{alg:exploration:sgd}).
We then update the confidence bounds on the estimated performance for the new policy $\policy_{w_t}$ and the deployed policy $\policy_d$ (Lines~\ref{alg:exploration:newlcb}--\ref{alg:exploration:baselinelcb}).
When the estimated lower bound performance of $\policy_{w_t}$ is better than the estimated upper bound performance of $\policy_d$ (Line~\ref{alg:exploration:comparison}), we deploy $\policy_{w_t}$, such that future actions are executed by the newly learned policy instead of the previous deployed policy $\policy_d$.

\subsection{Efficient policy evaluation}
\label{sec:efficient}

To implement \ac{SEA} we use the off-policy estimator described in Section~\ref{sec:offpolicy}. Such an implementation will be computationally expensive, because at each round $\nriterations$ we have to compute a sum over all $t$ data points collected so far. E.g., in Equation~\ref{eq:ch}, we need to compute the mean
\[
\frac{1}{t} \sum_{i=1}^t \hat{R}_i
\]
and the variance 
\[
\frac{1}{t} \sum_{i,j=1}^t \left(\hat{R}_i - \hat{R}_j\right)^2,
\]
at every round $t$.
Therefore, the complexity of applying Equation~\ref{eq:ch} or \ref{eq:chupper} would be $\mathcal{O}\left(T\right)$ and the complexity of Algorithm~\ref{alg:exploration} would be $\mathcal{O}\left(T^2\right)$.

Recall that $\mathcal{D}_t = \{(\mathbf{x}_i, a_i, r_{i,a_i}, p_i)\}_{i=1}^t$ is the collected log data and that there are $|\mathcal{A}|$ possible actions and $|\mathcal{X}|$ possible contexts.
We can then compute the mean $\frac{1}{t}\sum_{i=1}^t \hat{R}_i$ as follows (similar results hold for computing the variance term):
\begin{flalign}
\frac{1}{t} \sum_{i=1}^t \frac{r_{i,a_i}}{p_i} \pi(a_i \mid \bm{x}_i) &= \frac{1}{t} \sum_{a \in \mathcal{A}} \sum_{\bm{x} \in \mathcal{X}} \sum_{i=1}^t \mathbf{1}\left[ a_i = a \land \bm{x}_i = \bm{x} \right] \frac{r_{i,a_i}}{p_i} \pi(a \mid \bm{x}) \\
&=\frac{1}{t} \sum_{a \in \mathcal{A}} \sum_{\bm{x} \in \mathcal{X}} \pi(a \mid \bm{x}) \sum_{i=1}^t \mathbf{1}\left[ a_i = a \land \bm{x}_i = \bm{x} \right] \frac{r_{i,a_i}}{p_i},
\end{flalign}
where $\mathbf{1}[\cdot]$ is the indicator function.

The key insight here is that the inner sum, $\sum_{t=1}^{T} \mathbf{1}\left[ a_t = a \land \bm{x}_t = \bm{x} \right] \frac{r_{t,a_t}}{p_t}$ can be efficiently computed online during data collection and is independent of the policy $\pi_w$ that is being evaluated.
Specifically, we can create a zero-initialized matrix $W \in \mathbb{R}^{|\mathcal{A}| \times |\mathcal{X}|}$ where, each time a new data point $(\bm{x}_i, a_i, r_{i,a_i}, p_i)$ is logged, we update an entry as follows:
$$
W_{a_i,\bm{x}_i} \leftarrow W_{a_i,\bm{x}_i} + \frac{r_{i,a_i}}{p_i}.
$$
This results in the following method for computing the mean:
\begin{equation}
\frac{1}{t} \sum_{a \in \mathcal{A}} \sum_{x \in \mathcal{X}} \pi(a \mid x) W_{a,x} = \frac{1}{t} \sum_{(a, x) : W_{a,x} \neq 0} \pi(a \mid x)W_{a, x},
\end{equation}
which can be orders of magnitude faster to compute when $T \gg |\mathcal{A}| \cdot |\mathcal{X}|$.
This is not unreasonable in practice, as the size of interaction logs is usually much larger than the number of users and items.
In addition, we only need to compute the above sum for $a$ and $x$ for which $W_{a,x} \neq 0$.
This enables the use of sparse data structures that can speed up computation even further.

\subsection{Proof of safety}
\label{sec:proof}

\ac{SEA} is an \emph{online} learning method that we claim to be \emph{safe}.
To show that \ac{SEA} is actually safe, we bound the probability that a suboptimal policy, one whose expected reward is lower than the expected reward of the deployed policy, will be deployed during the duration of the algorithm.

\begin{theorem}
	\label{thm:safe}
	At any time $t$, with probability at least $1 - 2\delta$, \ac{SEA} will not deploy a suboptimal policy.
\end{theorem}

\begin{proof}
For notational simplicity we write the expected reward of a policy $\pi$ as:
\begin{equation}
R_t(\pi) = \mathbb{E}_{a\sim\pi}\left[\frac{1}{t} \sum_{i=1}^t r_{i,a}\right].
\end{equation}
First, from Algorithm~\ref{alg:exploration} we know that we only deploy $\pi_{w_t}$ when $\mathit{LCB}(\pi_{w_t}) > \mathit{UCB}(\pi_d)$.
Suppose that our confidence bounds do not fail.
In this case, it is impossible to deploy a suboptimal policy, since
\begin{equation}
R_t(\pi_{w_t})
\geq 
\mathit{LCB}(\pi_{w_t}) > \mathit{UCB}(\pi_d)
\geq
R_t(\pi_d)
\end{equation}
Consequently, a suboptimal policy can only be deployed when either the confidence bound estimate on our newly learned policy $\pi_{w_t}$ or our deployed policy $\pi_{d}$ fails.

The high-confidence off-policy estimators (see Equations~\ref{eq:ch} and~\ref{eq:chupper}) provide a lower and upper bound on the estimated performance.
According to Theorem 1 in \cite{thomas2015high}, these confidence bounds hold with probability at least $1 - \delta$.
Conversely, this means that the confidence bounds may fail with probability at most $\delta$:
\begin{equation}
P\left(R_t(\pi) \notin \left[\mathit{LCB}(\pi), \mathit{UCB}(\pi)\right]\right) \leq \delta,
\end{equation}
and, as a result, the probability that either the confidence bound on $\pi_{w_t}$ or $\pi_{d}$ will fail is at most $2\delta$:
\begin{equation}
P\left(R_t(\pi_{w_t}) \notin \left[\mathit{LCB}(\pi_{w_t}), \mathit{UCB}(\pi_{w_t})\right] \lor R_t(\pi_{d}) \notin \left[\mathit{LCB}(\pi_{d}), \mathit{UCB}(\pi_{d})\right]\right) \leq 2\delta.
\end{equation}
Therefore, with probability at least $1 - 2\delta$, we will not deploy a suboptimal policy.
\end{proof}

In our experiments (see Section~\ref{sec:experiments}) we set $\delta = 0.05$, which means that the algorithm is safe with probability at least $1 - 2\delta = 0.90$ at any time $t$.
We observe that this lower bound is fairly loose because we do not observe a single suboptimal deployment across many repetitions and possible deployment moments.

\subsection{Analysis}
\label{sec:analysis}
Conservative Linear Contextual Bandits (CLUCB)~\cite{kazerouni2016conservative} is, to our knowledge, the only other online learning method for contextual bandits that is also safe.
Unfortunately, CLUCB comes with two limitations:
\begin{enumerate}
	\item CLUCB does not scale beyond toy problems because it constructs confidence sets around parameters, which requires solving a constraint optimization problem every time an action has to be chosen which makes it infeasible for realistic \ac{IR} problems; and
	\item CLUCB can only be applied to linear models.
\end{enumerate}
Our method addresses both limitations.
First, \ac{SEA} scales to large and complex datasets, because it merely needs to compute the lower confidence bound and perform a gradient update step, both of which can be done highly efficiently.
Second, \ac{SEA} can easily be adapted to non-linear models such as gradient-boosted decision trees and neural networks as it is based on gradient descent.

\ac{SEA}, being an online learning method, is capable of exploration.
In contrast, \emph{offline} methods do not explore, they merely observe what a baseline policy is doing.
If this baseline policy does not explore well, offline learning techniques, whilst still able to learn, are less effective~\cite{swaminathan2015counterfactual}.
And even if the baseline policy is highly explorative, what truly matters is how well this policy explores the regions with favorable losses~\cite{owen2013monte}. 
\ac{SEA} solves this problem as it starts exploring actions using the newly learned policy as soon as it is safe to do so.
Since the policy learned by \ac{SEA} has a higher estimated performance than the baseline policy, the actions that it eventually takes are likely to be actions with high reward.

We note that the safety that \ac{SEA} guarantees does come at a cost: \ac{SEA} cannot guarantee that it will explore new actions beyond the initial deployed policy.
In contrast, purely \emph{online} methods can explore without any restrictions and as a result they may end up learning a better policy than \ac{SEA}.
Put differently, \ac{SEA} provides a trade-off between safety and exploration: with a lower $\delta$, \ac{SEA} will be more safe, but gives up some amount of exploration.
Vice versa, a higher $\delta$ allows \ac{SEA} to explore more aggressively while giving up some level of safety.


\section{Experimental Setup}
\label{sec:experiments}

To answer our research questions, we consider two tasks: \emph{text classification} and \emph{document ranking}.
We consider these two complementary tasks as they differ in the size of the action space, which is small in the case of text classification (the number of classes) but large in the case of document ranking (the number of possible ranked lists).
In both cases we turn a supervised learning problem into a bandit problem.

\subsection{Text classification task}

For \emph{text classification} we use a dataset $\mathcal{D} = \{(\contextvector_\iteration, \labels_\iteration)\}_{\iteration=1}^\nriterations$, where $\contextvector_\iteration \in \real^\contextdimensions$ is a feature representation of the object we wish to classify and $\labels_\iteration \in \{0, \ldots, \nractions\}$ is the correct label for that object.
E.g., in text classification $\contextvector_\iteration$ is a bag-of-words representation of a document we wish to classify and $\labels_\iteration$ the correct label for that document.
We are interested in the contextual bandit formulation for multi-class classification, i.e., where the correct label $\labels_\iteration$ is not known, but we can observe a reward signal for a chosen label $\action_\iteration$.

We follow the methodology of~\citet{beygelzimer2009offset} to transform a supervised learning problem into a contextual bandit problem:
\begin{enumerate}
	\item The environment presents the policy with the feature vector $\contextvector_\iteration$ (e.g., the bag-of-words representation of a text document).
	\item The policy chooses $\action_\iteration$, one of $\nractions$ possible labels, as its prediction.
	\item The environment returns a reward of 1 if the chosen label was correct and a reward of 0 if it was incorrect.
\end{enumerate}
Methods that use counterfactual learning (\ac{IPS} and \ac{SEA}), require a stochastic baseline policy and the corresponding propensity scores, i.e., the probability that the baseline policy chose the selected label $\action_\iteration$.
To this end we define our data-collection policy using the $\epsilon$-greedy approach.
An $\epsilon$-greedy approach offers several benefits over alternative exploration strategies:
\begin{enumerate}
\item The amount of exploration can be manually tuned with the $\epsilon$ parameter, allowing either very conservative or very aggressive exploration.
\item It is a priori known, via the $\epsilon$ parameter, how much performance we are giving up to perform exploration.
\item The $\epsilon$-greedy strategy generates stable and bounded propensity scores;
this counteracts the high variance problem that is common in counterfactual learning and evaluation.
\end{enumerate}
We consider two types of reward signals.
First, the perfect scenario where a reward of $1$ is given if the chosen label $a_t$ is correct and 0 otherwise.
Second, a near-random scenario where rewards are sampled from a Bernoulli distribution.
More specifically, when a policy chooses the correct label $a_t$, the policy is given a reward of $1$ with probability $0.6$.
Conversely, if the chosen label $a_t$ is incorrect, the learner is given a reward of $1$ with probability $0.4$.
We choose these contrasting scenarios as they highlight the need for safety.
It is much easier for a learner to make mistakes in the near-random scenario due to the high levels of noise and we hypothesize that safety plays a more important role in that case.
The text classification datasets that we use are specified in Table~\ref{tbl:classificationdatasets}.

\begin{table}
	\caption{Datasets used for the text classification task.}
	\label{tbl:classificationdatasets}
	\begin{tabular}{lcccc}
	\toprule
	\bf Dataset         & \bf Classes &\bf Features	& \bf Train size	& \bf Test size \\ \midrule
	\datasetusps~\cite{hull1994database}         &    10      & \phantom{00,}256       &  \phantom{1}7,291      &  \phantom{00}2,007 \\
	\datasetnews~\cite{lang1995newsweeder}    &    20     &  62,060   &    15,935  &  \phantom{00}3,993 \\
	\datasetrcv~\cite{lewis2004rcv1} & 53 & 47,236 & 15,564 & 518,571 \\
	\bottomrule
	\end{tabular}
\end{table}

\subsection{Document ranking task}
\label{sec:documentranking}
For the \emph{document ranking} task we use the counterfactual \ac{LTR} framework as described in~\cite{joachims2017unbiased}.
In this setup, a production ranker displays a ranked list to the user.
It is assumed that the user does not examine all documents in the presented ranking, but is instead more likely to observe top-ranked documents than lower ranked ones, a phenomenon called ``position bias''~\citep{joachims-2005-accurately}.
After a user observes a document they can either judge it as relevant, by clicking on it, or judge it as non-relevant, by not clicking on it.

We simulate user behavior as follows:
\begin{enumerate}
\item The policy being learned presents a ranked list to the simulated user.
\item The simulated user samples a set of observed documents from the ranked list, where the probability of observing the document at rank $i$ is 
\begin{equation}
p_i = \left( \frac{1}{i} \right)^\eta,
\label{eqn:doc_propensity}
\end{equation}
where $\eta$ is a parameter that controls the severity of click bias. This setup is identical to the one described in~\cite{joachims2017unbiased}.
\item For each observed document, we generate clicks depending on the relevance of the document. We use varying levels of click noise (``perfect'', ``position-biased'' and ``near-random''), the specific click probabilities are listed in Table~\ref{tbl:clicknoise}. This setup is in line with the methodologies described in~\cite{hofmann-balancing-2011,oosterhuis2016probabilistic,oosterhuis-balancing-2017,hofmann2013fidelity}.
\end{enumerate}
In our setup, clicks are simulated only on the top 10 documents. This more realistically simulates the behavior of Web search users who are unlikely to visit the second (or later) result page~\citep{chuklin-modeling-2013}.
For the counterfactual models we assume the correct propensity scores are known during learning.
In other words, the propensity model used to simulate position bias is also used to compute the propensity scores during learning.

We update our model via stochastic gradient descent at the document level by updating the weights using a weighted gradient and stochastic gradient descent, as described in Section~\ref{sec:learning}.
Table~\ref{tbl:rankingdatasets} details the datasets used for the document ranking task.

\begin{table}
	\caption{Datasets used for the document ranking task.}
	\label{tbl:rankingdatasets}
	\begin{tabular}{lccc}
	\toprule
	\bf Dataset         & \bf Queries & \mbox{}\hspace*{-1mm}\bf Features & \bf Avg. docs per query \\ \midrule
	MSLR-WEB10K~\cite{DBLP:journals/corr/QinL13}         &    10,000            &  136 & 124     \\
	Yahoo Webscope~\cite{chapelle2011yahoo}         &    36,251            &  700 & \phantom{0}23     \\
	Istella-s~\cite{lucchese2016post}    &    33,118        &    220 &  103 \\
	\bottomrule
	\end{tabular}
\end{table}

\begin{table}
	\caption{Click noise settings for the document ranking task. Each entry is the probability of clicking a document given its relevance label.}
	\label{tbl:clicknoise}
	\begin{tabular}{lrrrrr}
	\toprule
	 & \multicolumn{5}{c}{$P(\mathit{click} \mid \mathit{relevance})$} \\
	 \cmidrule{2-6}
	 & 0 & 1 & 2 & 3 & 4 \\
	 \midrule
	Perfect & 0.00 & 0.20 & 0.40 & 0.80 & 1.00 \\ 
	Position-biased & 0.10 & 0.10 & 0.10 & 1.00 & 1.00 \\ 
	Near random & 0.40 & 0.45 & 0.50 & 0.55 & 0.60 \\ 
	\bottomrule
	\end{tabular}
\end{table}

\subsection{Methods used for comparison}
Our experiments are aimed at assessing the performance of \ac{SEA} for solving contextual bandit problems. 
We compare \ac{SEA} against offline and online methods. 
The former comparison is motivated by the fact that offline methods are safe by definition so a comparison of \ac{SEA} against such methods will inform us about the potential performance gains by using \ac{SEA}.
The latter is motivated by the fact that \ac{SEA} is an online method so a comparison of \ac{SEA} against such methods will inform us about the safety of using \ac{SEA}.
In all experiments, the models to be trained are warm-started with the weights of the deployed policy.

\emph{Offline} methods for contextual bandit problems optimize a policy by observing actions that are taken by a separate logging policy $\policy_b$.
This makes them safe, because learning happens only on data that has been collected in the past by a logging policy.
However, safety also makes them incapable of exploration.
We use the state-of-the-art offline method $\lambda$-translated Inverse Propensity Scoring ($\lambda$-IPS)~\cite{joachimsdeep}.
We hypothesize that the average reward of a model trained by \ac{SEA} is higher than of a model trained by $\lambda$-IPS because \ac{SEA} is capable of exploring actions from the newly learned policy.

\emph{Online} methods optimize a policy $\policy_w$ by having it interact with the environment. 
Such methods are effective at exploring the action space but have no notion of safety. 
The online methods we use for the classification task are:
\begin{enumerate}
	\item policy gradient with an $\epsilon$-greedy strategy,
	\item policy gradient with a Boltzmann exploration strategy,
	\item LinUCB~\cite{li2010contextual}, and
	\item Thompson Sampling ~\cite{agrawal2013thompson}.
\end{enumerate}
For the ranking task, we use the following online methods:
\begin{enumerate}
\item SVMRank Online~\cite{joachims2002optimizing}
\item Dueling Bandit Gradient Descent~\cite{yue2009interactively} with Team-Draft Interleaving~\cite{radlinski2008does}
\end{enumerate}
We hypothesize that the user experience of online methods suffers in the early stages of learning (i.e., performance will be below the baseline policy $\policy_b$) because these approaches do not provide formal guarantees of safety while exploring new actions.

\emph{Safe} online methods are meant to be safe in the sense that during training, the algorithms never perform worse than a baseline policy $\policy_b$.
Our contribution, \ac{SEA}, falls in this class of methods.
We also consider CLUCB (Conservative Linear UCB)~\cite{kazerouni2016conservative}.
Its design makes it impractical for problems with high dimensionality or large action spaces as it requires performing an $\contextdimensions \times \contextdimensions$ matrix inversion and solving a constrained optimization problem whenever an action has to be selected;
its high computational complexity prevents us from using it with our datasets.

Finally, we also include an empirically safe version of \ac{SEA}, which does not use the high-confidence bounds, which we denote as \acf{BSEA}.
This method compares the mean estimate of the reward of the learned policy $\policy_w$ and the deployed policy $\policy_d$, instead of the lower and upper confidence bounds.
Specifically, we use Algorithm~\ref{alg:exploration} where we set
\begin{equation}
\ucb(\policy, \mathcal{D}_t) = \lcb(\policy, \mathcal{D}_t) = \frac{1}{t} \sum_{i=1}^{t} \frac{r_{i,a_i}}{p_i}\policy(a_i \mid \bm{x}_i).
\end{equation}
We expect this method to deploy its learned model earlier and more frequently than \ac{SEA} because it does not have to overcome potentially large confidence bounds.
Hence, we expect this policy to do better than \ac{SEA}, but it may exhibit unsafe behavior as it no longer provides the same formal safety guarantees as \ac{SEA}.

\subsection{Choice of baseline policy}
Both $\lambda$-IPS and \ac{SEA} depend on a baseline policy $\policy_b$.
We require that this baseline policy is sub-optimal, so that learning can occur.
This requirement is reasonable since we cannot hope to improve an already optimal policy.
We introduce suboptimality in $\policy_b$ by subsampling the training set on which we train the policy (1\% sample for the classification task and 0.1\% sample for the ranking task).
This is motivated by a scenario that commonly occurs in real search engines or classification systems:
Manual labels are expensive to obtain and are usually available on a scale of several orders of magnitude smaller than logged bandit feedback.
This strategy for introducing suboptimality results in a baseline policy whose performance is much better than random, but still not optimal.

\subsection{Metrics and statistical significance}
We evaluate \ac{SEA} and the competing approaches in two ways. 
One is in terms of \emph{cumulative reward} during training, which is the sum of all rewards received as a function of the number of rounds; this type of metric allows us to quantify the degree to which the user experience is affected.
A higher cumulative reward during training indicates a better user experience.
The second is in terms of \emph{average reward}, which is the reward averaged per action on held-out test data; this type of metric allows us to quantify how well a trained policy generalizes to unseen test data.
The document ranking task uses simulated clicks as a reward signal (see Section~\ref{sec:documentranking}) during training, which is not a very insightful metric for evaluating the true performance of a policy.
Instead, to \emph{evaluate} the learned policies, we use nDCG@10~\cite{jarvelin2002cumulated}.
We measure statistical significance of observed differences using a $t$-test ($p < 0.01$).


\section{Results}
\label{section:results}

\subsection{Safety}
\label{sec:results:safety}
We first address \textbf{RQ1}: 
\begin{quote}
\textit{\rqone}
\end{quote}
To answer RQ1, we plot the performance of \ac{SEA} against that of \emph{online} algorithms while they are training, using average reward on a held-out test dataset as our metric.

Let us first consider the text classification task; see Figure~\ref{fig:classification:online}.
The baseline policies have an average accuracy of around 0.6, except in the \datasetnews~dataset where this is 0.4.
The tested online algorithms include $\epsilon$-greedy, Boltzmann exploration, Thompson sampling and LinUCB.
For the \datasetnews~and \datasetrcv~datasets we do not run LinUCB and Thompson sampling because they require inverting a $\numprint{47236} \times \numprint{47236}$ matrix (for \datasetnews) and $\numprint{62060} \times \numprint{62060}$ matrix (for \datasetrcv) on every update which is too computationally expensive.
The performance of \ac{SEA} is at least as good as the baseline policy.
Similarly, it seems that warm-started online algorithms are also safe in this scenario, as they too perform always at least as good as the baseline policy.
The best algorithm in terms of final performance is different for each of the datasets.
Overall we find that when it is computationally feasible to apply UCB or Thompson sampling, they outperform all other methods.
Furthermore we find that Boltzmann exploration works very well with a perfect reward signal.
However, in the case of near random rewards, both Boltzmann exploration and $\epsilon$-greedy are not always capable of learning a good policy.
Finally, we observe that \ac{BSEA} is on par with \ac{SEA}, and in some cases significantly outperforms it, while being empirically safe.
\ac{BSEA} is capable of deploying its learned model faster and more frequently than \ac{SEA} as it does not have to overcome potentially large confidence bounds.
Although this means that \ac{BSEA} is not guaranteed to be safe, we find that it is empirically safe across all experimental conditions.

Next, we consider the document ranking setting; see Figure~\ref{fig:ranking:online}.
We use NDCG@10 on held-out test data as the evaluation metric.
We observe that \ac{SEA} is always at least as good as the baseline policy whereas online learning methods may suffer from unsafe performance in the early stages of learning.
We see that in the case of near-random clicks on the Istella-s dataset, \ac{SEA} accurately identifies that it is not safe to deploy whereas the online method suffers from unsafe performance and actually go below the baseline level.
However, on the other datasets with near-random clicks, \ac{SEA} never deploys a new model and consequently does not explore different actions.
As a result, it is unable to improve upon the production policy.
\emph{Online} approaches outperform \ac{SEA} here because they explore more aggressively and are able to learn even in the face of large amounts of noise.
These results are in line with previous work~\cite{jagerman-2019-model}, which has shown that counterfactual methods have difficulty learning in scenarios with large amounts of noise.
In cases with position-biased or perfect clicks, the ranker trained by SEA performs on par with the online method.
This means that in realistic scenarios we do not sacrifice any performance by using \ac{SEA}.
Finally, we note that \ac{BSEA} performs empirically safe on the document ranking task across all experimental settings, which is in line with the results on the text classification task.

This answers RQ1. Our experiments indicate that \ac{SEA} is indeed safe, its performance does not fall below that of the baseline policy.

\begin{figure*}
	\includegraphics[width=0.67\textwidth]{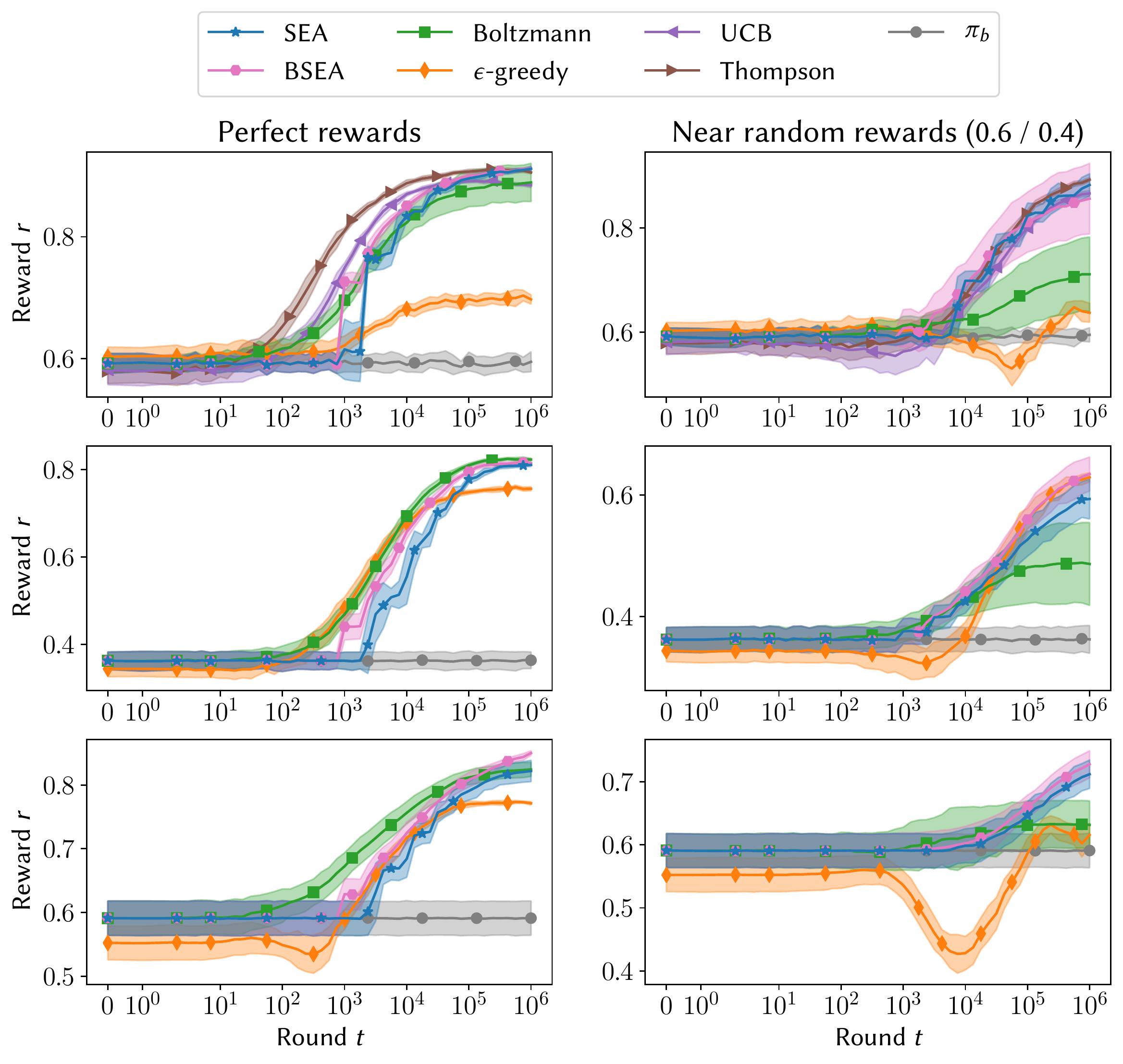}
	
	\caption{The performance of \ac{SEA} (blue line) compared to \emph{online} algorithms for the text classification task. 
		The top row indicates the USPS dataset, the middle row the 20 Newsgroups dataset and the bottom row the RCV1 dataset.
		The performance is measured on a held-out test set after each round $t$. 
		Shaded areas indicate standard deviation. Best viewed in color.}
	\label{fig:classification:online}
\end{figure*}

\begin{figure*}[h]
	\includegraphics[width=\textwidth]{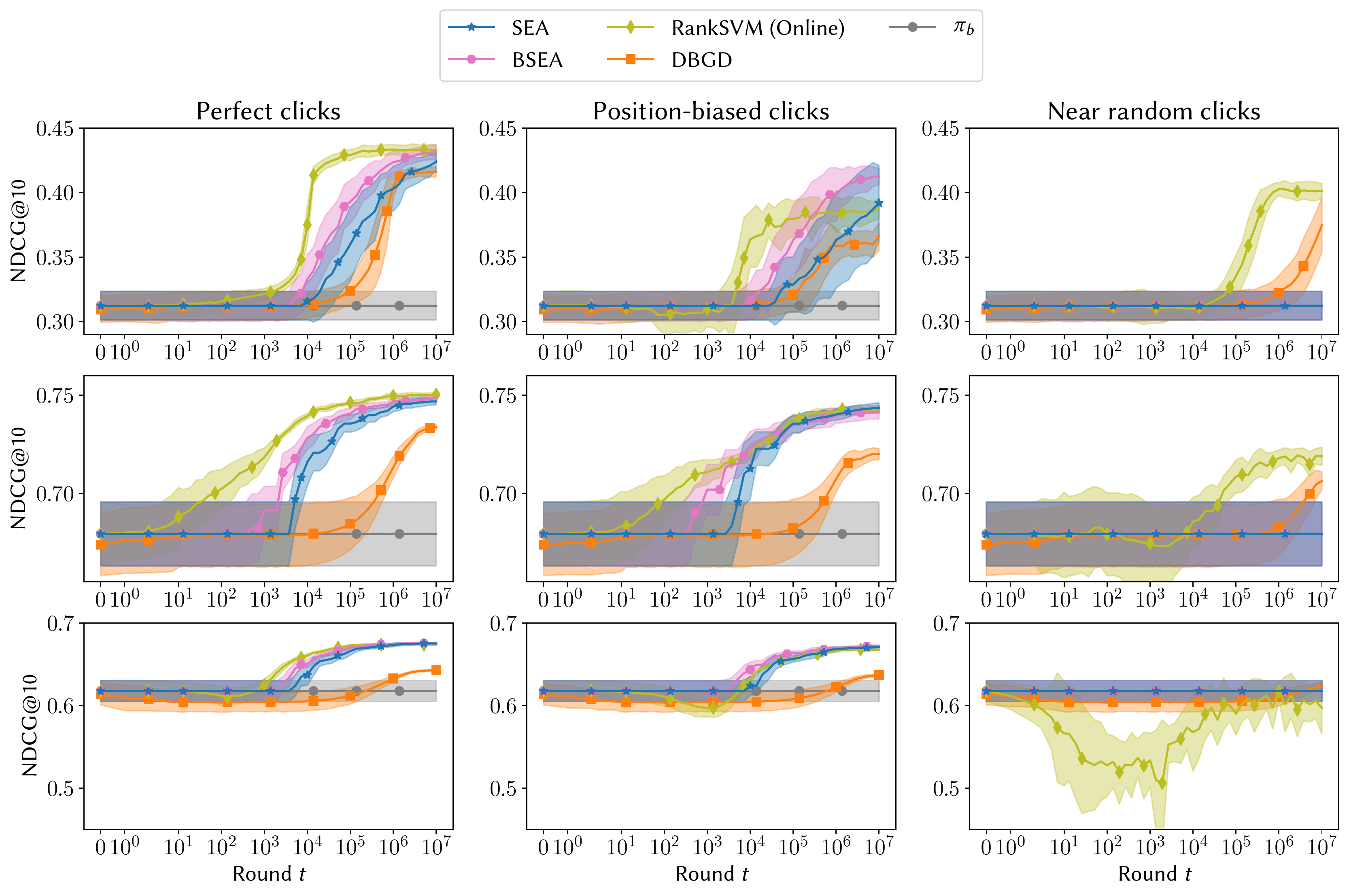}
	
	\caption{The performance of \ac{SEA} (blue line) compared to an \emph{online} approach (orange line) for the document ranking task on MSLR-10k (top row), Yahoo Webscope (middle row) and Istella-s (bottom row) with varying levels of click noise.
	The performance of the trained ranker is measured on a held-out test set after each round $t$. 
	Shaded areas indicate standard deviation. Best viewed in color.}
	\label{fig:ranking:online}
\end{figure*}

\subsection{Improved user experience}
\label{sec:results:userexperience}
Next, we answer \textbf{RQ2}: 
\begin{quote}
\textit{\rqtwo}
\end{quote}
To answer RQ2, we measure the cumulative reward obtained by the learning algorithms.
During training, the offline method, $\lambda$-IPS,  only observes actions taken by the baseline policy, hence its cumulative reward is always equal to the cumulative reward of the baseline and is omitted from the result tables to save space.

Table~\ref{tbl:classification:results} lists the results for the text classification task.
The cumulative reward achieved by \ac{SEA} is at least as good as the baseline, and eventually better.
We find that only $\epsilon$-greedy on the RCV1 dataset under near random rewards performs significantly worse than the baseline.
In all other settings we find that all the warm-started online methods similarly provide a user experience that is at least as good as the baseline, and eventually better.
This means that for the text classification task, all methods, with the exception of $\epsilon$-greedy, perform safely and never harm the user experience.
Finally, we see that \ac{BSEA} is able to improve the user experience faster and more quickly than \ac{SEA} in the perfect rewards setting, and performs on par in the near random rewards setting.
Similar to the results we have observed in Section~\ref{sec:results:safety}, we hypothesize that \ac{BSEA} does not have to overcome potentially large confidence bounds and as a result is able to more quickly and more frequently deploy its learned model.
As a result, \ac{BSEA} is capable of improving the user experience over \ac{SEA}.

Next we turn to the case of ranking. See Table~\ref{tbl:ranking:results} for the cumulative reward results for document ranking.
It is clear that \ac{SEA} performs at least as good as the baseline policy and eventually outperforms it.
For the online learning method, this is not the case.
Specifically for the Istella-s dataset with near random clicks, we observe significant performance degradations in the early stages of learning which harms the user experience.
However, on the other datasets with near random clicks we find that the online methods are capable of safely exploring the action space, even with high noise, and as a result improve the overall user experience more than \ac{SEA}.
From Table~\ref{tbl:dist:rel} we see that Istella-s contains very sparse relevance feedback compared to the MSLR10k and Webscope datasets. In the near-random click scenario this means that there are significantly more clicks on non-relevant documents, making Istella-s with near-random clicks the most challenging learning scenario.
We once again see that \ac{BSEA} is a very strong baseline across all datasets, improving the user experience over \ac{SEA} in the perfect and position-biased clicks settings while performing on par with \ac{SEA} in the near-random click setting.

This answers RQ2.
\ac{SEA} provides a better user experience than online methods, particularly in the document ranking setting with high levels of click noise where the user experience may be negatively affected by online algorithms.
In the later stages of learning \ac{SEA} provides a user experience that is significantly better than the baseline and comparable to online approaches.
We find that \ac{BSEA} is able to improve the user experience even further by being able to deploy its learned model earlier and more frequently.

\begin{table*}
	\caption{Cumulative reward, relative to the baseline policy, while training for text classification. LinUCB and Thompson Sampling cannot be run on the 20 Newsgroups and RCV1 datasets due to their high computational complexity. 
	Statistically significant differences with \ac{SEA} are indicated using $^\blacktriangle$ ($p < 0.01$) for gains and  $^\blacktriangledown$ ($p < 0.01$) for losses.}
	\label{tbl:classification:results}

\begin{tabular}{ll@{\hspace{1cm}}rrrrr}
  \toprule
    & \bf Round $t$         & $10^{2}$ & $10^{3}$ & $10^{4}$ & $10^{5}$ & $10^{6}$ \\
  \midrule
    \multicolumn{7}{c}{\emph{Perfect rewards}}\\
  \midrule
    \multirow{6}{*}{\rotatebox[origin=c]{90}{USPS}}
      & UCB               & 1.73  & 105.53\rlap{$^\blacktriangle$}  & 2568.33\rlap{$^\blacktriangle$}  & 32496.67\rlap{$^\blacktriangle$}  & 350226.27\rlap{$^\blacktriangle$} \\
      & Thompson          & 4.47  & 172.40\rlap{$^\blacktriangle$}  & 2833.00\rlap{$^\blacktriangle$}  & 33298.33\rlap{$^\blacktriangle$}  & 352940.07\rlap{$^\blacktriangle$} \\
      & $\epsilon$-greedy & 0.80  & 31.87  & 699.07\rlap{$^\blacktriangledown$}  & 10568.53\rlap{$^\blacktriangledown$}  & 116854.20\rlap{$^\blacktriangledown$} \\
      & Boltzmann         & 0.93  & 69.27\rlap{$^\blacktriangle$}  & 1950.47  & 28002.60  & 313452.47 \\
      & {BSEA             } & {0.00}  & {0.00}  & {2086.93\rlap{$^\blacktriangle$}}  & {30483.40\rlap{$^\blacktriangle$}}  & {341919.60\rlap{$^\blacktriangle$}} \\
      & SEA               & 0.00  & 0.00  & 1720.73  & 29286.93  & 337577.53 \\[2ex]
    \multirow{4}{*}{\rotatebox[origin=c]{90}{20-News}}
      & $\epsilon$-greedy & -2.00  & 60.07\rlap{$^\blacktriangle$}  & 2571.53\rlap{$^\blacktriangle$}  & 42913.73\rlap{$^\blacktriangle$}  & 499605.53\rlap{$^\blacktriangledown$} \\
      & Boltzmann         & 0.33  & 56.47\rlap{$^\blacktriangle$}  & 2388.53\rlap{$^\blacktriangle$}  & 45220.73\rlap{$^\blacktriangle$}  & 564285.20\rlap{$^\blacktriangle$} \\
      & {BSEA             } & {0.00}  & {0.00}  & {1935.87\rlap{$^\blacktriangle$}}  & {43151.33\rlap{$^\blacktriangle$}}  & {564200.73\rlap{$^\blacktriangle$}} \\
      & SEA               & 0.00  & 0.00  & 1083.47  & 37342.53  & 540070.47 \\[2ex]
    \multirow{4}{*}{\rotatebox[origin=c]{90}{RCV1}}
      & $\epsilon$-greedy & -3.53  & -28.33  & 905.40\rlap{$^\blacktriangle$}  & 21396.73\rlap{$^\blacktriangle$}  & 283490.73 \\
      & Boltzmann         & 0.80  & 51.40  & 1356.07\rlap{$^\blacktriangle$}  & 22085.53\rlap{$^\blacktriangle$}  & 269428.07 \\
      & {BSEA             } & {0.00}  & {0.00}  & {946.67}  & {22094.07\rlap{$^\blacktriangle$}}  & {322481.00\rlap{$^\blacktriangle$}} \\
      & SEA               & 0.00  & 0.00  & 573.40  & 18856.20  & 282126.73 \\
  \midrule
    \multicolumn{7}{c}{\emph{{Near random rewards (0.6 / 0.4)}}}\\
  \midrule
    \multirow{6}{*}{\rotatebox[origin=c]{90}{{USPS}}}
      & {UCB              } & {-1.73}  & {3.93}  & {89.00}  & {3617.53}  & {59778.47} \\
      & {Thompson         } & {-1.73}  & {-3.47}  & {121.67}  & {4024.93}  & {61446.93} \\
      & {$\epsilon$-greedy} & {-0.60}  & {9.73}  & {52.87}  & {-534.00\rlap{$^\blacktriangledown$}}  & {7380.40\rlap{$^\blacktriangledown$}} \\
      & {Boltzmann        } & {0.33}  & {3.20}  & {52.00}  & {1297.73\rlap{$^\blacktriangledown$}}  & {23223.27\rlap{$^\blacktriangledown$}} \\
      & {BSEA             } & {0.00}  & {0.00}  & {87.53}  & {3687.80}  & {51662.27} \\
      & {SEA              } & {0.00}  & {0.00}  & {50.40}  & {3567.00}  & {55458.73} \\[2ex]
    \multirow{4}{*}{\rotatebox[origin=c]{90}{{20-News}}}
      & {$\epsilon$-greedy} & {-1.87}  & {-11.33}  & {-20.67}  & {2440.47}  & {56843.47} \\
      & {Boltzmann        } & {-2.40}  & {-2.67}  & {95.40}  & {2183.27}  & {32669.13} \\
      & {BSEA             } & {-0.33}  & {1.60}  & {108.67}  & {3015.47}  & {58141.20} \\
      & {SEA              } & {-0.33}  & {1.60}  & {82.53}  & {2534.20}  & {48437.53} \\[2ex]
    \multirow{4}{*}{\rotatebox[origin=c]{90}{{RCV1}}}
      & {$\epsilon$-greedy} & {-0.47}  & {-9.40}  & {-253.80\rlap{$^\blacktriangledown$}}  & {-1362.07\rlap{$^\blacktriangledown$}}  & {5222.00\rlap{$^\blacktriangledown$}} \\
      & {Boltzmann        } & {-0.33}  & {0.80}  & {27.53}  & {807.40}  & {8887.93\rlap{$^\blacktriangledown$}} \\
      & {BSEA             } & {0.00}  & {0.47}  & {0.93}  & {905.87}  & {25312.67} \\
      & {SEA              } & {0.00}  & {0.00}  & {7.60}  & {706.60}  & {21776.53} \\
  \bottomrule
\end{tabular}

\end{table*}

\begin{table*}
	\caption{Cumulative NDCG@10, relative to the baseline policy, while training for document ranking under varying levels of click noise. Statistical significance is denoted in the same way as in Table~\ref{tbl:classification:results}.}
	\label{tbl:ranking:results}

\begin{tabular}{ll@{\hspace{1cm}}rrrrrr}
  \toprule
    & \bf Round $t$         & $10^{2}$ & $10^{3}$ & $10^{4}$ & $10^{5}$ & $10^{6}$ & $10^{7}$ \\
  \midrule
    \multicolumn{8}{c}{\emph{Perfect clicks}}\\
  \midrule
    \multirow{4}{*}{\rotatebox[origin=c]{90}{MSLR10k}}
     & RankSVM (Online)  & 0.43  & 6.48  & 259.78\rlap{$^\blacktriangle$}  & 10065.58\rlap{$^\blacktriangle$}  & 115823.32\rlap{$^\blacktriangle$}  & 1180033.12\rlap{$^\blacktriangle$} \\
     & {DBGD             } & {-0.05}  & {-0.62}  & {-1.22}  & {541.43}  & {48244.85\rlap{$^\blacktriangle$}}  & {940724.90\rlap{$^\blacktriangle$}} \\
     & {BSEA             } & {0.00}  & {0.00}  & {49.80}  & {5513.30\rlap{$^\blacktriangle$}}  & {95560.77\rlap{$^\blacktriangle$}}  & {1125341.57\rlap{$^\blacktriangle$}} \\
     & SEA               & 0.00  & 0.00  & 6.52  & 2724.04  & 71507.83\rlap{$^\blacktriangle$}  & 996331.68\rlap{$^\blacktriangle$} \\[2ex]
    \multirow{4}{*}{\rotatebox[origin=c]{90}{Webscope}}
     & RankSVM (Online)  & 1.81  & 31.75\rlap{$^\blacktriangle$}  & 514.23\rlap{$^\blacktriangle$}  & 6477.67\rlap{$^\blacktriangle$}  & 70834.57\rlap{$^\blacktriangle$}  & 730703.57\rlap{$^\blacktriangle$} \\
     & {DBGD             } & {-0.14}  & {-1.22}  & {-4.49}  & {222.10}  & {19343.31}  & {475561.27\rlap{$^\blacktriangle$}} \\
     & {BSEA             } & {-0.01}  & {2.98}  & {325.39\rlap{$^\blacktriangle$}}  & {5518.60\rlap{$^\blacktriangle$}}  & {64913.61\rlap{$^\blacktriangle$}}  & {687719.69\rlap{$^\blacktriangle$}} \\
     & SEA               & -0.01  & -0.05  & 151.24  & 4661.96\rlap{$^\blacktriangle$}  & 60199.38\rlap{$^\blacktriangle$}  & 670419.85\rlap{$^\blacktriangle$} \\[2ex]
    \multirow{4}{*}{\rotatebox[origin=c]{90}{Istella-s}}
     & RankSVM (Online)  & 0.05  & 0.96  & 325.50\rlap{$^\blacktriangle$}  & 5364.95\rlap{$^\blacktriangle$}  & 61424.60\rlap{$^\blacktriangle$}  & 631101.11\rlap{$^\blacktriangle$} \\
     & {DBGD             } & {-1.27}  & {-13.16}  & {-130.50}  & {-911.40}  & {6156.68}  & {222276.99\rlap{$^\blacktriangle$}} \\
     & {BSEA             } & {0.00}  & {0.00}  & {206.87}  & {4749.66\rlap{$^\blacktriangle$}}  & {59609.20\rlap{$^\blacktriangle$}}  & {634840.28\rlap{$^\blacktriangle$}} \\
     & SEA               & 0.00  & 0.00  & 76.25  & 3982.61\rlap{$^\blacktriangle$}  & 56417.23\rlap{$^\blacktriangle$}  & 624784.49\rlap{$^\blacktriangle$} \\
  \midrule
    \multicolumn{8}{c}{\emph{Position-biased clicks}}\\
  \midrule
    \multirow{4}{*}{\rotatebox[origin=c]{90}{MSLR10k}}
     & RankSVM (Online)  & -0.23  & -5.17  & 184.43  & 5845.83\rlap{$^\blacktriangle$}  & 68623.01\rlap{$^\blacktriangle$}  & 725200.78\rlap{$^\blacktriangle$} \\
     & {DBGD             } & {-0.06}  & {-0.81}  & {-3.12}  & {368.26}  & {30558.36\rlap{$^\blacktriangle$}}  & {453895.21\rlap{$^\blacktriangle$}} \\
     & {BSEA             } & {0.00}  & {0.00}  & {3.85}  & {2849.70}  & {69915.76\rlap{$^\blacktriangle$}}  & {927569.46\rlap{$^\blacktriangle$}} \\
     & SEA               & 0.00  & 0.00  & -0.01  & 856.81  & 32375.36  & 634675.69\rlap{$^\blacktriangle$} \\[2ex]
    \multirow{4}{*}{\rotatebox[origin=c]{90}{Webscope}}
     & RankSVM (Online)  & 1.04  & 26.88  & 349.04\rlap{$^\blacktriangle$}  & 4960.61\rlap{$^\blacktriangle$}  & 61055.11\rlap{$^\blacktriangle$}  & 642520.10\rlap{$^\blacktriangle$} \\
     & {DBGD             } & {-0.15}  & {-1.15}  & {-6.12}  & {96.66}  & {14787.60}  & {353452.62\rlap{$^\blacktriangle$}} \\
     & {BSEA             } & {0.00}  & {8.86}  & {324.38\rlap{$^\blacktriangle$}}  & {4920.24\rlap{$^\blacktriangle$}}  & {57553.64\rlap{$^\blacktriangle$}}  & {614338.66\rlap{$^\blacktriangle$}} \\
     & SEA               & 0.00  & 0.01  & 144.06  & 4497.80\rlap{$^\blacktriangle$}  & 57702.23\rlap{$^\blacktriangle$}  & 633904.88\rlap{$^\blacktriangle$} \\[2ex]
    \multirow{4}{*}{\rotatebox[origin=c]{90}{Istella-s}}
     & RankSVM (Online)  & -0.08  & -13.49  & -14.85  & 3303.76\rlap{$^\blacktriangle$}  & 50451.05\rlap{$^\blacktriangle$}  & 561270.85\rlap{$^\blacktriangle$} \\
     & {DBGD             } & {-1.27}  & {-13.14}  & {-133.17}  & {-1132.55}  & {-1749.76}  & {150641.69} \\
     & {BSEA             } & {0.00}  & {0.00}  & {128.66}  & {4145.78\rlap{$^\blacktriangle$}}  & {54614.90\rlap{$^\blacktriangle$}}  & {603275.09\rlap{$^\blacktriangle$}} \\
     & SEA               & 0.00  & 0.00  & 0.00  & 3165.49\rlap{$^\blacktriangle$}  & 50018.04\rlap{$^\blacktriangle$}  & 581799.13\rlap{$^\blacktriangle$} \\
  \midrule
    \multicolumn{8}{c}{\emph{Near random clicks}}\\
  \midrule
    \multirow{4}{*}{\rotatebox[origin=c]{90}{MSLR10k}}
     & RankSVM (Online)  & 0.11  & -0.71  & -24.48  & 850.82  & 66464.76\rlap{$^\blacktriangle$}  & 849713.67\rlap{$^\blacktriangle$} \\
     & {DBGD             } & {-0.06}  & {-0.78}  & {-7.48}  & {-6.20}  & {4004.35}  & {339609.56\rlap{$^\blacktriangle$}} \\
     & {BSEA             } & {0.00}  & {-0.00}  & {-0.01}  & {-0.06}  & {-0.09}  & {0.84} \\
     & SEA               & 0.00  & -0.00  & -0.01  & -0.06  & -0.09  & 0.84 \\[2ex]
    \multirow{4}{*}{\rotatebox[origin=c]{90}{Webscope}}
     & RankSVM (Online)  & 0.23  & -2.44  & 1.81  & 1829.45  & 34646.76\rlap{$^\blacktriangle$}  & 403139.88\rlap{$^\blacktriangle$} \\
     & {DBGD             } & {-0.14}  & {-1.24}  & {-7.46}  & {-50.84}  & {1100.02}  & {169029.43} \\
     & {BSEA             } & {0.01}  & {0.03}  & {0.03}  & {-0.12}  & {0.83}  & {1.18} \\
     & SEA               & 0.01  & 0.03  & 0.03  & -0.12  & 0.83  & 1.18 \\[2ex]
    \multirow{4}{*}{\rotatebox[origin=c]{90}{Istella-s}}
     & RankSVM (Online)  & -7.09  & -83.71\rlap{$^\blacktriangledown$}  & -678.03\rlap{$^\blacktriangledown$}  & -3210.06\rlap{$^\blacktriangledown$}  & -17794.95  & -138288.71 \\
     & {DBGD             } & {-1.27}  & {-13.07}  & {-135.09}  & {-1284.11}  & {-10421.04}  & {3837.78} \\
     & {BSEA             } & {0.00}  & {0.00}  & {0.00}  & {0.00}  & {0.00}  & {0.00} \\
     & SEA               & 0.00  & 0.00  & 0.00  & 0.00  & 0.00  & 0.00 \\
\bottomrule
\end{tabular}

\end{table*}

\begin{table*}
	\caption{Distribution of relevance grades on documents for the ranking datasets. Note that Istella-s is very sparse, nearly 90\% of the documents are judged as non-relevant.}
	\label{tbl:dist:rel}

\begin{tabular}{l@{\hspace{1cm}}rrrrr}
  \toprule
   {Relevance grade} & {0} & {1} & {2} & {3} & {4} \\
  \midrule
    {Webscope} & {0.26} & {0.36} & {0.28} & {0.08} & {0.02} \\
    {MSLR10k} & {0.52} & {0.32} & {0.13} & {0.02} & {0.01} \\
    {Istella-s} & {0.89} & {0.02} & {0.04} & {0.03} & {0.02} \\
  \bottomrule
\end{tabular}

\end{table*}

\subsection{The benefit of exploring}
\label{sec:results:exploration}
Finally, we turn to \textbf{RQ3}: 
\begin{quote}
\textit{\rqthree}
\end{quote}
To answer RQ3, we compare \ac{SEA} to the state-of-the-art offline learning algorithm $\lambda$-IPS on unseen test data.
The results for the classification task are displayed in Table~\ref{tbl:final:classification}.
\begin{table*}
	\caption{Average reward on held-out test data for the learned model for the document classification task after 1,000,000 rounds. Statistical significance is denoted the same way as in Table~\ref{tbl:classification:results}.}
	\label{tbl:final:classification}

\begin{tabular}{ll@{\hspace{1cm}}rr}
  \toprule
     & 
 & \emph{Perfect rewards  } & \emph{{Near random rewards (0.6 / 0.4)}}\\
  \midrule
    \multirow{7}{*}{\rotatebox[origin=c]{90}{USPS}}
      & $\epsilon$-greedy & 0.90\rlap{$^\blacktriangledown$}  & {0.82\rlap{$^\blacktriangledown$}} \\
      & Boltzmann & 0.90  & {0.72\rlap{$^\blacktriangledown$}} \\
      & UCB & 0.91\rlap{$^\blacktriangledown$}  & {0.91} \\
      & Thompson & 0.91\rlap{$^\blacktriangledown$}  & {0.91} \\
      & IPS & 0.92  & {0.89} \\
      & {BSEA} & {0.92}  & {0.86} \\
      & SEA & 0.92  & {0.90} \\[2ex]
    \multirow{5}{*}{\rotatebox[origin=c]{90}{20-News}}
      & $\epsilon$-greedy & 0.83  & {0.69} \\
      & Boltzmann & 0.88\rlap{$^\blacktriangle$}  & {0.53\rlap{$^\blacktriangledown$}} \\
      & IPS & 0.83\rlap{$^\blacktriangledown$}  & {0.53\rlap{$^\blacktriangledown$}} \\
      & {BSEA} & {0.86}  & {0.74} \\
      & SEA & 0.85  & {0.70} \\[2ex]
    \multirow{5}{*}{\rotatebox[origin=c]{90}{RCV1}}
      & $\epsilon$-greedy & 0.85  & {0.68} \\
      & Boltzmann & 0.83  & {0.63\rlap{$^\blacktriangledown$}} \\
      & IPS & 0.81  & {0.68} \\
      & {BSEA} & {0.87\rlap{$^\blacktriangle$}}  & {0.74} \\
      & SEA & 0.84  & {0.73} \\
  \bottomrule
\end{tabular}

\end{table*}
Because \ac{SEA} is capable of exploring, it finds highly favorable regions of the action space.
$\lambda$-IPS is, by design, incapable of exploration and cannot deviate far from the baseline policy in terms of performance.
As a result, in the case of text classification, the final model learned by \ac{SEA} outperforms the final model learned by $\lambda$-IPS on the 20 Newsgroups datasets.
On the USPS and RCV1 dataset, there is no noticeable performance difference between $\lambda$-IPS and \ac{SEA}.
Note that the baseline policy for USPS and RCV1 has a performance of around 0.6, whereas the baseline policy for the 20 Newsgroups dataset only has a performance of around 0.4.
We postulate that this difference in baseline performance may cause $\lambda$-IPS to learn better and therefore perform equally to \ac{SEA} on the USPS and RCV1 datasets.
Lastly, we find that \ac{BSEA} performs on par with \ac{SEA} in all settings, while producing a significantly better model on the RCV1 dataset with perfect rewards.
Similar to our observations in Sections~\ref{sec:results:safety} and~\ref{sec:results:userexperience}, we find that \ac{BSEA} likely performs so well because it does not have to overcome confidence bounds and can deploy its learned model faster and more frequently than \ac{SEA}, allowing it to learn a more effective model.

Finally, we consider the document ranking setting; see Table~\ref{tbl:final:ranking}.
\ac{SEA} is able to learn a more effective ranker than $\lambda$-IPS on all datasets with perfect clicks and on Istella-s with position-biased clicks.
We hypothesize that this is because \ac{SEA}, being capable of exploration, eventually shows documents to the user that the baseline policy would rarely, if ever, show.
As a result, \ac{SEA} is able to obtain clicks on these documents, allowing it to learn more effectively.
This is in line with our expectations because previous work has shown that even a tiny amount of exploration can result in substantial improvements in \ac{LTR}~\cite{hofmann-contextual-2011}.
Finally, we find that \ac{BSEA} does not produce a better ranker than \ac{SEA} for the document ranking task. Across all settings, the models produced by \ac{BSEA} and \ac{SEA} are comparable.

\begin{table*}
	\caption{Average reward on held-out test data for the learned model for the document ranking task after 10,000,000 rounds. Statistical significance is denoted the same way as in Table~\ref{tbl:classification:results}.}
	\label{tbl:final:ranking}

\begin{tabular}{ll@{\hspace{1cm}}rrr}
  \toprule
    & 
 & \emph{Perfect clicks   } & \emph{Position-biased clicks} & \emph{Near random clicks}\\
  \midrule
    \multirow{5}{*}{\rotatebox[origin=c]{90}{MSLR10k}}
      & RankSVM (Online) & 0.43  & 0.39  & 0.40\rlap{$^\blacktriangle$} \\
      & {DBGD} & {0.43}  & {0.38}  & {0.39\rlap{$^\blacktriangle$}} \\
      & IPS & 0.35\rlap{$^\blacktriangledown$}  & 0.35\rlap{$^\blacktriangledown$}  & 0.32 \\
      & {BSEA} & {0.43}  & {0.41}  & {0.32} \\
      & SEA & 0.43  & 0.40  & 0.32 \\[2ex]
    \multirow{5}{*}{\rotatebox[origin=c]{90}{Webscope}}
      & RankSVM (Online) & 0.75  & 0.74  & 0.72\rlap{$^\blacktriangle$} \\
      & {DBGD} & {0.74\rlap{$^\blacktriangledown$}}  & {0.72\rlap{$^\blacktriangledown$}}  & {0.71} \\
      & IPS & 0.73\rlap{$^\blacktriangledown$}  & 0.74\rlap{$^\blacktriangledown$}  & 0.69 \\
      & {BSEA} & {0.75}  & {0.74}  & {0.69} \\
      & SEA & 0.75  & 0.74  & 0.69 \\[2ex]
    \multirow{5}{*}{\rotatebox[origin=c]{90}{Istella-s}}
      & RankSVM (Online) & 0.67  & 0.67\rlap{$^\blacktriangledown$}  & 0.60 \\
      & {DBGD} & {0.66\rlap{$^\blacktriangledown$}}  & {0.65\rlap{$^\blacktriangledown$}}  & {0.64} \\
      & IPS & 0.67\rlap{$^\blacktriangledown$}  & 0.66\rlap{$^\blacktriangledown$}  & 0.64 \\
      & {BSEA} & {0.68}  & {0.67}  & {0.64} \\
      & SEA & 0.68  & 0.67  & 0.64 \\
  \bottomrule
\end{tabular}

\end{table*}

This answers RQ3.
Exploration does indeed help the performance of the policy learned by \ac{SEA} and it outperforms $\lambda$-IPS both for scenarios with a small action space (e.g., text classification) and for scenarios with a large action space (e.g., document ranking).


\section{Related work}

The idea of deploying automated decision making systems in \ac{IR} is not new. 
The contextual bandit framework has been used in news recommendation~\citep{li2010contextual}, ad placement~\citep{langford2008exploration}, and online learning to rank~\citep{hofmann-balancing-2011}. 
Contextual bandit formulations of \ac{IR} problems allow insights and methods from the bandit literature to be applied and extended to address these problems~\citep{hofmann-contextual-2011}.
For research on contextual bandits this connection has opened up an application area where new
approaches can be evaluated on large-scale datasets~\citep{li2011unbiased}.

\subsection{Learning in contextual bandits}
Contextual bandit algorithms have been widely studied in the \emph{online} and \emph{offline} learning settings~\cite{agrawal2013thompson,li2010contextual,swaminathan2015counterfactual}.
A key challenge in contextual bandit problems is the exploration-vs-exploitation tradeoff.
On the one hand we want to explore new actions so as to find favorable rewards.
On the other hand, we want to exploit existing knowledge about actions so as to maximize the total reward.
The online learning methods we consider in this paper deal with this tradeoff in various ways.
The methods we consider are policy gradient with $\epsilon$-greedy exploration, policy gradient with Boltzmann exploration, LinUCB and Thompson Sampling.
Policy gradient methods optimize the weights of a policy via stochastic gradient descent, by solving the following optimization problem:
\begin{equation}
\min_{\policy_w} - \sum_{t=1}^{T} \log \left( \policy_w(\action_\iteration \mid \contextvector_\iteration) \right) \cdot \reward_\iteration.
\end{equation}
The $\epsilon$-greedy heuristic selects actions by choosing with probability $\epsilon$ an action uniformly at random and with probability $(1 - \epsilon)$ the best possible action.
Boltzmann exploration~\cite{cesa2017boltzmann} chooses actions by drawing them with a probability proportional to the policy: $\action_\iteration \sim \policy_w(~\cdot \mid \contextvector_\iteration)$.
LinUCB~\cite{li2010contextual} and Thompson Sampling~\cite{agrawal2013thompson} are different from policy gradient methods because they make a linearity assumption.
These methods construct a set of weights $w_a \in \real^m$ for every possible action, such that $\policy_w(\action_\iteration \mid \contextvector_\iteration) = w_{\action_\iteration}^T \contextvector_\iteration$.
LinUCB selects actions by choosing the one with the highest confidence bound.
Thompson Sampling samples new weights $\hat{w}_{\action}$ from a posterior distribution and then chooses the best action given the sampled weights: $\action_\iteration = \arg\max_a \pi_{\hat{w}_{\action}}(\action \mid \contextvector_\iteration)$.

In \ac{IR} there has been considerable attention for exploiting log data~\citep{hofmann-reusing-2013}. 
It is one of the most ubiquitous forms of data available, as it can be recorded from a variety of systems  at little cost~\citep{swaminathan-batch-2015}.
The interaction logs of such systems typically contain a record of the input to the system, the prediction made by the system, and the feedback. The feedback provides only partial information limited to the particular prediction shown by the system.
Offline learning algorithms tell us how data collected from interaction logs of one system can be used to optimize a new system~\cite{swaminathan-batch-2015,hofmann-contextual-2011}.
Interaction logs used in offline learning are usually biased towards the policy that collected the data.
To remove this bias, offline learning methods resort to inverse propensity scoring.
To use inverse propensity scoring, the logging policy must be stochastic and the corresponding propensity scores (probability of choosing the logged action) are recorded.
Using these propensity scores, it is possible to reweigh data samples to remove the bias.
An advantage of offline learning methods is that they do not require interactive experimental control.
Thus, there is no risk of hurting the user experience with these methods.
In other words, offline learning methods are safe by definition~\citep{jagerman-2019-model}.

Unlike \emph{online} methods, \ac{SEA}'s performance during the early stages of learning does not suffer and always stays at least as good as a baseline, making the method safe to use.
Compared to \emph{offline} methods, \ac{SEA} is capable of exploration, which makes it effective at finding areas of the action space that may have high reward.

\subsection{Safety in contextual bandits} 
There has been a growing interest in concepts related to safety for contextual bandit problems.
This is due to the fact that contextual bandit formulations are applied to automated decision making systems, where actions taken by the system can have a significant impact on the real world.
There are two main groups of work:  \emph{risk-aware} methods and \emph{conservative} methods.

Risk-aware methods~\cite{galichet2013exploration,huo2017risk,sun2016risk} are online learning algorithms that model the risk associated with executing certain actions as a cost which is to be constrained and minimized.
\citet{galichet2013exploration} introduce the concept of risk-awareness for the multi-armed bandit framework.
\citet{sun2016risk} extend the idea of risk-awareness to the adversial contextual bandit setting.
\citet{garcia2012safe} explore risk-awareness for the reinforcement learning paradigm.
Risk-aware methods use a separate type of feedback signal, in addition to the standard reward feedback, which is called \emph{risk}. Risk-aware methods aim to keep the cumulative risk below a specific threshold.
These types of methods are typically applied in fields like robotics where certain actions can be dangerous or cause damage.
A drawback is that the risk has to be explicitly quantified by the environment.
Designing a good risk feedback signal for \ac{IR} systems is subject to modeling biases and in some cases impossible, limiting the application of such methods.

Conservative methods~\cite{wu2016conservative} measure safety as a policy's performance relative to a baseline policy. 
A method is safe if its performance is always within some margin of the baseline policy.
The idea was first introduced for the multi-armed bandit case~\cite{wu2016conservative} and later extended to linear contextual bandits in the form of the CLUCB algorithm~\cite{kazerouni2016conservative}.
CLUCB is a safe conservative online learning method, which works by constructing confidence sets around the parameters of the policy.
Unfortunately, the method has to solve a constrained optimization problem every time an action has to be selected, which has a significant computational overhead.
In contrast, \ac{SEA} addresses this problem because it only needs to do two computationally efficient operations: update a lower confidence bound estimate and perform a gradient update step.
This makes \ac{SEA} applicable to larger and more complex datasets.

Finally, \citet{li-2019-bubblerank} introduce a complementary approach to \ac{SEA}, called BubbleRank, an algorithm that gradually improves upon an initial ranked list by exchanging higher-ranked less attractive items for lower- ranked more attractive items. \citeauthor{li-2019-bubblerank} define a safety constraint that is based on incorrectly-ordered item pairs in the ranked list, and prove that BubbleRank never violates this constraint with a high probability. We do not compare to BubbleRank in our experiments because it assumes a user model~\citep{chuklin-click-2015}, an assumption that is orthogonal to our experimental setup~\citep{jagerman-2019-model}.


\section{Conclusion}

We have proposed \ac{SEA}, a \emph{safe} \emph{online} learning algorithm for contextual bandit problems.
\ac{SEA} learns a new policy from the behavior of an existing baseline policy and then starts to execute actions from the new policy once its estimated performance is at least as good as that of the baseline.
This brings us the best of both worlds, achieving the performance of \emph{online} learning with the safety of \emph{offline} learning.

We perform extensive experimentation on two \ac{IR} tasks, \emph{text classification} and \emph{document ranking}.
In both tasks \ac{SEA} is safe.
It never performs worse than a baseline policy, whereas online methods are unsafe and suffer from suboptimal performance in the early stages of learning.
We observe that the user experience with \ac{SEA} is improved in the early stages of training, but may be suboptimal in later stages when compared to methods that converge much faster, such as LinUCB (although such methods are not generalizable to all datasets).
The final performance of a model learned with \ac{SEA} is comparable to other online algorithms and beats that of offline methods, which are incapable of exploration.
Finally, we find that \ac{BSEA}, a boundless version of \ac{SEA}, is empirically just as safe as \ac{SEA} while being able to explore faster and, as a result, outperform \ac{SEA} in many experimental conditions.
These results confirm that SEA does indeed trade off advantages and disadvantages of \emph{offline} and \emph{online} learning, in some scenarios outperforming online methods in the early stages of learning and having higher final performance than offline methods.
However, compared to purely \emph{online} approaches, \ac{SEA} may not be as effective in learning a good policy due to the fact that it is more conservative when exploring.
The conservative nature of \ac{SEA} implies that safety is not free, there is a possible performance cost involved for cases where \ac{SEA} is unable to effectively explore.

An interesting direction for future work is an extension of \ac{SEA} to non-linear models. \ac{SEA} builds on gradient descent and it is trivial to extend the method to use gradient-boosted decision trees or neural networks.
This line of work is especially applicable for the document ranking task where it is known that non-linear models can outperform linear models by a wide margin~\cite{wu2010adapting}.
Furthermore, another possibility for future work is to perform a study on safety for a broad range of online and offline methods, similar to~\cite{jagerman-2019-model}.
Specifically, it would be interesting to compare against recent work on safe online learning to re-rank~\cite{li-2019-bubblerank}.

\section*{Code and data}
To facilitate reproducibility of the results in this paper, we are sharing all resources used in this paper at \url{https://github.com/rjagerman/tois2019-safe-exploration-algorithm}.
This includes the training procedures and parameters of the online, offline, safe and baseline policies.

\begin{acks}
We thank Chang Li and Harrie Oosterhuis for valuable discussions.
We thank our anonymous reviewers for inspiring questions and helpful suggestions.
        
This research was partially supported by
Ahold Delhaize,
the Association of Universities in the Netherlands (VSNU),
the Innovation Center for Artificial Intelligence (ICAI),
and
the Netherlands Organization for Scientific Research (NWO) under project number 612.\-001.\-551.

All content represents the opinion of the authors, which is not necessarily shared or endorsed by their respective employers and/or sponsors.
\end{acks}
        
\bibliographystyle{ACM-Reference-Format}
\bibliography{paper} 

\end{document}